\documentclass[a4paper]{article}

\usepackage[english]{babel}
\usepackage[utf8]{inputenc}
\usepackage{amsmath}
\usepackage{graphicx}
\usepackage[colorinlistoftodos]{todonotes}
\usepackage{stmaryrd} 
\usepackage{amsthm}
\usepackage{amssymb}
\newtheorem{theorem}{Theorem}
\theoremstyle{definition}
\newtheorem{definition}{Definition}

\usepackage{color}
\usepackage{listings}
\usepackage{caption}

\lstnewenvironment{algorithm}[1][] 
{
\captionsetup{labelsep=colon} 
\lstset{ 
    mathescape=true,
    frame=tB,
    numbers=left,
    numberstyle=\tiny,
    basicstyle=\scriptsize,
    keywordstyle=\color{black}\bfseries\em,
    keywords={,input, note, output, return, datatype, function, in, if, else, foreach, while, begin, end, }
    numbers=left,
    xleftmargin=.04\textwidth,
    #1 
}
}
{}
\usepackage[]{hyperref}
\hypersetup{
    pdftitle={Your title here},
    pdfauthor={Your name here},
    pdfsubject={Your subject here},
    pdfkeywords={keyword1, keyword2},
    bookmarksnumbered=true,
    bookmarksopen=true,
    bookmarksopenlevel=1,
    colorlinks=true,
    pdfstartview=Fit,
    pdfpagemode=UseOutlines,    
    pdfpagelayout=TwoPageRight
}

\title{Semantic Query Language for Temporal Genealogical Trees}

\author{Evgeniy Gryaznov}

\date{\today}

\begin{document}
\maketitle

\renewcommand{\abstractname}{}    
\begin{abstract}
Computers play a crucial role in modern ancestry management, they are used to collect, store, analyze, sort and display genealogical data. However, current applications do not take into account the kinship structure of a natural language.

In this paper we propose a new domain-specific language KISP which is based on a formalization of English kinship system, for
accessing and querying traditional genealogical trees. KISP is a dynamically typed LISP-like programming language with a rich set
of features, such as kinship term reduction and temporal information expression.

Our solution provides a user with a coherent genealogical framework that allows for a natural navigation over any traditional family tree.
\end{abstract}

\section{Introduction}

With the advent of computers, we are able to manage genealogies of incredible size. Now, due to the progress in storage engineering, it became possible to maintain and expand existing ancestries of considerable size. But merely keeping data on physical disks is not enough. To sufficiently realize the full potential of computers in genealogy management, one should also provide means of inquiry in ancestral data.

There is another important concept to consider when working with family trees, namely, the concept of time. A genealogy can exist only in specific temporal framework that is imposed on it by the very nature of history itself. As a result, any computer representation of an ancestry that lacks this framework is exorbitantly inadequate. Therefore, its preservation is a crucial feature for any software that is aimed for effective genealogy management.

If we want to teach computers understand lineage, we need to construct some type of artificial language that will allow us to effectively navigate and query any possible family tree. But observe, that an ancestry already has its own idiosyncratic terminology and grammar, which can be successfully utilized as a natural basis for such a language. Our research is an attempt to do exactly that.

\section{Related Work}
Ontologies find their natural application in the context of our work.
Since the original formulation of a concept, a lot of software has been developed to manage ontologies, including such systems as
Protege, Inge and others. These systems have already been heavily used in the variety of different fields.

For instance, Tan Mee Ting \cite{tanfamily} designed and
implemented a genealogical ontology using Protege and evaluated its consistency with Pellet, HermiT and FACT++ reasoners. He
showed that it is possible to construct a family ontology using \textit{Semantic Web} \cite{web} technologies with full capability of
exchanging family history among all interested parties. However, he did not address the issue of navigating the family tree using
kinship terms.

Ontological can be used to model any kind of family tree, but the problem arises when a user wants to query his relatives using
kinship terms. No standard out-of-the-box ontological query language is able to articulate statements such as in our example above.
Although an ontology can be tailored to do so, it is not in any way a trivial task. Maarten Marx \cite{xcpath} addressed
this issue, but in the different area. He designed an extension for XPath, the first order node-selecting language for XML.

Catherine Lai and Steven Bird \cite{ling} described the domain of linguistic trees and discussed the expressive requirements for a
query language. They presented a language that can express a wide range of queries over these trees, and showed that the
language is first order complete. This language is also an extension of XPath.

Artale et al. \cite{artale} did a comprehensive survey of various temporal knowledge representation formalisms.
In particular, they analysed ontological and query languages based on the linear
temporal logic LTL, the multi-dimensional Halpern-Shoham interval temporal logic, as well as metric temporal logic (MTL). They
noted that the W3C standard ontological languages, OWL 2 QL and OWL 2 EL, are designed to represent knowledge over a static domain,
and are not well-suited for temporal data.

Modelling kinship with mathematics and programming languages, such as LISP, has been an extensive area of research. Many people
committed a lot of work into the field, including Bartlema and Winkelbauer, who investigated \cite{bartlema} structures of a
traditional family and wrote a simple program that assigns fathers to children. Their main purpose was to understand how this
structure affects fertility, mortality and nuptiality rates. Although promising, small steps has been made towards designing a language to reason about kinship. Also, their program cannot express temporal information.

Another prominent attempt in modelling kinship with LISP was made \cite{findler} by Nicholas Findler, who examined various kinship structures and combined them together to create a LISP program that can perform arbitrary complex kinship queries. Although his solution is culture-independent, he did not take the full advantage of LISP as a programming language, and because of that it is impossible to express queries which are not about family interrelations. In contrast, our system does not suffer from that restriction.

More abstract, algebraic approach was taken \cite{read} by D. W. Read, who analysed the terminology of American Kinship in terms
of its mathematical properties. His algebra clearly demonstrates that a
system of kin terms obeys strict rules which can be successfully ascertained by formal methods. In another article \cite{read2} he
discusses how software, in the broader sense of intelligent systems, can help anthropologists understand foreign cultures.

Periclev et al. developed \cite{vlad} a LISP program called KINSHIP that produces the guaranteed-simplest analyses, employing a minimum number of features and components in kin term definitions, as well as two further preference constraints that they propose in their paper, which reduce the number of multiple component models arising from alternative simplest kin term definitions conforming to one feature set. The program is used to study the morphological and phonological properties of kin terms in English and Bolgarian languages.

According to authors \cite{read} and \cite{vlad} there have been many attempts to create adequate models of kinship based on
mathematics and computational systems. For instance, a prominent french mathematician Andre Weil analysed \cite{weil} the Murngin
system of kinship and marriage using \textit{group theory}. In particular, he showed how one can embed the nuptial rules of a
particular society into the framework of permutation groups $S_n$. His work was later extended by Bush \cite{bush}, who proposed
the concept of \textit{permutation matrices} as a more effective tool for analysis. Kemeny, Snell and Thompson were \cite{kemeny} first to
systematize the properties of prescriptive marriage systems as an integrated set of axioms. All distinct kinship structures which
satisfy these axioms were systematically derived and described by White \cite{white}, whose more practical \textit{generators} set
the structural analysis on a more concrete basis.

Similar attempt was made by John Boyd, who also used \cite{boyd} the apparatus of group theory to give a mathematical
characterization of the conditions under which groups become relevant for the study of kinship. He argued that the concept of
group extension and its specialization to direct and semi-direct products determine the evolutionary sequences and the coding of
these kinship systems.

Ernest Gellner discusses \cite{gellner}, from the pure philosophical point of view, the possibility of constructing an \textit{ideal
language} for an arbitrary kinship structure.

However, despite remarkable progress, there is a certain doubt in the mathematical community about the applicability of such
abstract approaches to the study of kinship. For example, White \cite{white} recognized the failure of his structural analysis of
societies like Purum or Murngin, which practice matrilateral cross-cousin marriage. Liu addresses \cite{liu} this problem with
establishing a new mathematical method for the analysis of prescriptive marriage systems.

We also note that none of the works include the temporal element into their formalisms, which provides novelty for our research.

\section{Formal Language of Kinship}
The study of kin structures has its roots in the field of anthropology. Among the first foundational works was Henry Morgan's \textit{magnum opus} "Systems of Consanguinity and Affinity of the Human Family"\cite{morgan}, in which he argues that all human societies share a basic set of principles for social organization along kinship\footnote{Recall that in this paper, the word "kinship" includes relatives as well as in-laws} lines, based on the principles of \textbf{consanguinity} (kinship by blood) and \textbf{affinity} (kinship by marriage).

Following Henry Morgan, we recognize two primary types of family bonds: marital (affinity) and
parental (consanguinity). These bonds define nine basic kin terms: \textit{father, mother, son, daughter, husband, wife, parent, child and spouse}. Observe that combining them in different ways will yield all possible kinship terms that can and do exist.

For instance, \textit{cousin} is \textit{a child of a child of a parent of a parent} of a particular person. Another example: \textit{mother-in-law} is just \textit{a mother of a spouse}.

Now let us represent a traditional Christian family tree as a special type of \textit{ontology} with its' own concepts,
attributes, relations and constraints. Concepts are people in a family, their attributes are: \textit{name, birth date,
birthplace, sex} and relations are parental and marital bonds with a wedding date.

Together with the everything stated above, we have the following cultural constraints imposed on our genealogy:
\begin{enumerate}
    \label{en:req}
    \item{Each person can have any finite number of children.}
    \item{Each person can have at most two parents of different sex.}
    \item{Each person can have at most one spouse of different sex.}
    \item{A spouse cannot be a \textit{direct relative}, i.e. a sibling or a parent. In other words, direct incest is
        prohibited.}
\end{enumerate}

When considering those prerequisites one should bear in mind that we deliberately focused only on rules, taboos and
customs of one particular culture, namely American culture in the sense of Read \cite{read}. Under different assumptions and
in further studies, these conditions can be relaxed and revisited.

Apart from these four, here are two additional temporal constraints that express the interrelation between birth and wedding
dates:
\begin{enumerate}
    \item{No one can marry a person before he or she was born, i.e. a wedding date can only be strictly after a
        birth date of each spouse.}
    \item{A parent is born strictly before all of his (her) children.}
\end{enumerate}
Due to the general nature of these two constraints, they are always true in every culture and therefore can be safely
assumed in our work.

Every genealogy that meets these six requirements we shall call a \textbf{traditional family tree}. As the name "tree" suggests,
we can indeed view this structure as a graph with its vertices as people and edges as bonds. Observe that every kinship term
corresponds exactly to a \textit{path} between ego and specified relative. Under such view, kin term becomes a set of instructions,
telling how to get from the starting vertex \texttt{A} to the end vertex \texttt{B}. For example, consider the term
\textit{mother-in-law}. What is it if not precisely a \textit{directive}: "firstly, go to my spouse, then proceed to her mother".
The wonderful thing is that, due to the nature of kinship terminology, we can \textit{compose} old terms together to create
new, even those which do not have their own name. This simple observation shows that we can see kinship terms as paths in a family
tree underpins our entire research.

Now, if we want to efficiently query a traditional family tree, we need to further investigate the mathematical features of
the language of kinship terms.

Here we present our attempt to model the language of traditional American, in the sense of Read\cite{read}, kinship terminology.
There are three main characteristics that define every formal language: its syntax (spelling, how words are formed), semantics
(what does particular word mean) and pragmatics (how a language is used).

\subsection{Syntax}
We use Backus-Naur Form to designate the syntax for our formal language. Let $\Sigma$ be the set of six basic kinship terms:
\textit{father, mother, son, daughter, husband, wife}. Then we can express the grammar as follows:
\begin{align*}
\label{al:syntax}
term ::= \Sigma | (term \cdot term) | (term \vee term) | (term)^{-1} | (term)^{\dagger}
\end{align*}
The first operation is called \textit{concatenation}, second -- \textit{fork}, third -- \textit{inverse} and the last --
\textit{dual}. We denote this language by $\mathcal{L}$.

Here are some examples of ordinary kinship terms expressed in our new language. Note that we deliberately omit superfluous
parentheses and the composition sign for the sake of simplicity:
\begin{itemize}
    \item{ Parent is $father \vee mother$. }
    \item{ Child is $son \vee daughter$. }
    \item{ Brother is $son (father \vee mother)$. }
    \item{ Sibling is $(son \vee daughter)(father \vee mother)$. }
    \item{ Uncle is $son(father \vee mother)(father \vee mother)$. }
    \item{ Daughter-in-law is $daughter \cdot husband$ }
    \item{ Co-mother-in-law is $mother(husband \vee wife)(son \vee daughter)$ }
\end{itemize}
From these examples you can see the real power of this language -- the power to express all possible
kinship terms.  Now the important step towards solving our main goal, developing a language for managing
temporal genealogies, is to assign meaning to these words. From now on we distinguish between \textit{artificial} kinship terms,
i.e. well-formed terms of our formalization, and \textit{natural} kin terms used in ordinary English. By referring to just terms,
we mean the former, if nothing else is stated.

\subsection{Semantics}
Let $\Sigma^*$ stand for the set of all possible kin terms generated from the basis $\Sigma$ using the previously defined syntax.
Let $\mathcal{G} = (V, E)$ be a traditional family tree with $V$ as a set of its vertices (people) and $E$ as a set of its
edges (bonds). Moreover, because $\mathcal{G}$ is traditional, every person from the set $V$ have the following attributes:
\begin{itemize}
    \item{A father. We will denote him as $father(p)$, a function that returns a \textit{set} containing at most one element.}
    \item{A mother. We will denote her by $mother(p)$.}
    \item{A set of his or her children: $children(p)$.}
    \item{A set of his or her sons: \[son(p) = \{c | c \in children(p) \land Male(p)\}\]}
    \item{A set of his or her daughters: \[daughter(p) = \{c | c \in children(p) \land Female(p)\}\]}
    \item{A spouse: $spouse(p)$.}
    \item{A husband: \[husband(p) = \{s | spouse(s) \land Male(p)\}\]}
    \item{A wife: \[wife(p) = \{s | spouse(s) \land Female(p)\}\]}
\label{it:basic-func}
\end{itemize}
Due to the constraints stated in \ref{en:req}, result-set of $father$, $mother$, $spouse$, $husband$ and $wife$ can contain at
most one element.

Now we are ready to introduce \textbf{Denotational Semantics} for $\Sigma^*$. This name was chosen because it highly resembles
its namesake semantics of programming languages. Note that we regard kinship terms as \textit{functions} on subsets of $V$. Each
function takes and returns a specific subset of all relatives, so its type is $f : \mathcal{P}(V) \to \mathcal{P}(V)$.

We proceed by induction on the syntactic structure of $\mathcal{L}$. Let $t$ be an element of $\Sigma^*$, then:
\begin{enumerate}
    \item{If $t \in \Sigma$, then $\llbracket t \rrbracket = F(t)$, where $F(t)$ assigns to each basic kin term its corresponding
        function from the list \ref{it:basic-func}.}
    \item{Term concatenation is a composition of two functions:
        \[\llbracket (t_1 \cdot t_2) \rrbracket = \llbracket t_1 \rrbracket \circ \llbracket t_2 \rrbracket\]}
    \item{Fork is a set-theoretic union of results of its sub-functions:
        \[\llbracket (t_1 \vee t_2) \rrbracket = p \mapsto \llbracket t_1 \rrbracket(p) \cup \llbracket t_2
    \rrbracket(p)\]}
    \item{Term inverse is exactly the inverse of its function:
        \[\llbracket t_1^{-1} \rrbracket = \llbracket t_1 \rrbracket^{-1}\]}
\end{enumerate}
The \textit{dual} operator ($\dagger$) is more difficult to define. We want it to mean exactly the same as the term, where the
gender of each its basic sub-term is reversed, e.g. dual of "uncle" is "aunt", dual of "brother" is "sister" and so on.
Here we can use induction once again:
\begin{enumerate}
    \item{If $t \in \Sigma$, then $\llbracket t \rrbracket = D(t)$, where $D(t)$ is a basic term of opposite sex.}
    \item{Dual is distributive over concatenation, i.e. dual of concatenation is a concatenation of duals:
        \[\llbracket (t_1 \cdot t_2)^\dagger \rrbracket = \llbracket (t_1^\dagger \cdot t_1^\dagger) \rrbracket\]}
    \item{Dual is distributive over forking:
        \[\llbracket (t_1 \vee t_2)^\dagger \rrbracket = \llbracket (t_1^\dagger \vee t_2^\dagger ) \rrbracket\]}
    \item{Inverse commutes with dual:
            \[\llbracket (t^{-1})^{\dagger} \rrbracket = \llbracket (t^{\dagger})^{-1} \rrbracket\]}
\end{enumerate}
Observe that we also have the distributivity of concatenation over forking.
This semantics allows us to efficiently navigate any family tree.

\subsection{Pragmatics}
Now, when syntax and semantics has been defined, let's discuss the applications of our new formalism.

One of the main goals in constructing our language was achieving \textit{cultural independence}. That is, the language should
assign a unique term to every relative in genealogy without relying on labels and kinship words from a particular society.
Although we use terms such as \texttt{father}, \texttt{daughter} and \texttt{husband} from English for the basis of the language,
we do that only for readability, since they can be easily replaced with abstract placeholders like $x$ or $y$.

The proposed formalism finds its natural application in the field of \textit{machine translation}. Languages drastically, and
often even incompatibly, differ in the way they express kinship information. The correct translation of kinship terms still poses
a challenge for linguists and anthropologists. For example, in Russian there is a word for a son of
a brother of ones' wife, but no corresponding term in English.
With our formalism we can encode the meaning of such words and use it to provide a more adequate translation between any possible
pair of natural languages.

We can also apply the formalism to the problem of \textit{cross-cousin marriage}: given the description of a particular society, a
genealogy of a family from that society and two individuals from it, ascertain whether the rules of their community allow them
to marry. This problem is extensively studied in the field of computational anthropology. What makes it especially
difficult is that each culture has its own peculiar set of regulations and laws regarding this subject.
In every case, one of the important steps towards solving the problem is to research and establish a correct model of that
society.  With our new formalism we can facilitate this process.

\section{Term Reduction}
\label{sec:reduc}
Our artificial language has a problem: its too verbose. Indeed, to encode such ubiquitous kin terms as "uncle" or "great-nephew" one
must use quite lengthy phrases that are hard to write and read. It is therefore important to have some sort of reduction mechanism
for our language that will shorten long terms into a small set of common kinship relations to help understand them.

Firstly, let us analyse the problem. We have the following mapping $\omega$ between $\Sigma^*$ and the set of English kinship terms
$\mathcal{W}$
\begin{align*}
    son(father \vee mother) &\mapsto \text{brother}\\
    daughter(father \vee mother) &\mapsto \text{sister}\\
    father(father \vee mother) &\mapsto \text{grandfather}\\
    mother(father \vee mother) &\mapsto \text{grandmother}\\
    son(son \vee daughter) &\mapsto \text{grandson}\\
    \vdots\\
    father(son \vee daughter)(wife \vee husband)(son \vee daughter) &\mapsto \text{co-father-in-law}
\end{align*}
This dictionary allows us to effectively translate between kin terms of our artificial language $\mathcal{L}$ and their usual
English equivalents. We can also view this mapping as a \textit{regular grammar} in the sense of Chomsky hierarchy\cite{chomsky}.
However, note that we strictly prohibit mixing these two collections and therefore we deliberately avoid
using words from the RHS in the LHS, because otherwise the grammar will lose its regularity and become \textit{at least}
context-free, making the problem even more challenging. Let us define another function on top of $\omega$ that will replace the
first sub-term $u \subset t$ in a term $t \in \Sigma^*$:
\begin{align*}
    \Omega_u(t) = t[u / \omega(u)]
\end{align*}
Here the only change in meaning of $t[u / \omega(u)]$ is that the substitution takes place only once.

Now the task can be stated thusly: given a term $t \in \Sigma^*$ find its shortest (in terms of the number of concatenations)
translation under $\Omega$, i.e. which sub-terms need to be replaced and in what order.

This problem can be easily reduced to that of of finding the desired point in the tree of all possible
substitutions. Moreover, this vertex is actually a \textit{leaf}, because otherwise it is not the shortest one. But the latter
can be solved by just searching for this leaf in-depth. Unfortunately, the search space grows exponentially with the
number of entries in the dictionary $\omega$, thus making the naive brute-force approach unfeasible.

Here we propose a heuristic greedy algorithm \ref{algo:red} that, although does not work for all cases, provides an expedient
solution to the reduction problem in $O(n^2)$ time. Firstly, it finds the longest sub-term $u$ that exists in the dictionary
$\omega$, then divides the term into two parts: left and right from $u$, after that it applies itself recursively to them, and
finally it concatenates all three sub-terms together.

Now let's analyze the time complexity of this algorithm:
\begin{theorem}
    The execution time of the algorithm listed in \ref{algo:red} belongs to $\Theta(n^2)$.
\end{theorem}
\begin{proof}
    Let $T(n)$ be the execution time of the algorithm, where $n$ stands for the number of concatenations in a kin term.
    First of all, observe that $T(n)$ obeys the following recurrence:
    \begin{align}
        T(n) = 2 T(n / 2) + O(n^2)
    \end{align}
    Indeed, we make a recursive call exactly \textit{two} times and each call receives roughly the half of the specified term.
    During execution the function passes through two nested cycles, so one call costs us $O(n^2)$.

    Secondly, to solve a recurrence, we use the \textbf{Master Method} from the famous book \textit{Introduction to Algorithms}\cite{cormen} by Cormen et al.
    In our case $a = 2$, $b = 2$, and $f(n) = O(n^2)$. Observe, that if we take $\epsilon$ to be any positive real number below
    one: $0 < \epsilon < 1$, then $f(n) = \Omega(n^{\log_ba + \epsilon}) = \Omega(n^{1 + \epsilon})$.

    Let us show that $f(n)$ satisfies the \textit{regularity} criterion: $af(n / b) \leqslant cf(n)$ for some constant $c < 1$.
    Indeed, just pick $c = 1/2$:
    \begin{align*}
        2f(n / 2) &\leqslant cf(n),\\
        2\frac{n^2}{4} &\leqslant cn^2,\\
        \frac{1}{2}n^2 &\leqslant cn^2,\\
        \frac{1}{2}n^2 &\leqslant \frac{1}{2}n^2
    \end{align*}
    Thus, we can use the third case from the Master Method, which tells us that $T(n) = \Theta(n^2)$.
\end{proof}

\subsection{Pursuing Confluence}
Current approach listed in \ref{algo:red} has one major disadvantage: like any other greedy algorithm it can fail to choose a
correct reduction path between two terms with equal amount of concatenations.
We can alleviate this by augmenting our rewriting system, based on $\omega$ dictionary, with a feature called \textit{confluence},
also known as \textit{Church-Rosser} property:
\begin{definition}
    An abstract term rewriting system is said to possess \textbf{confluence}, if, when two terms $N$ and $P$ can be yielded from $M$,
    then they can be reduced to the single term $Q$. Figure \ref{fig:conf} depicts this scenario.
\end{definition}
Not only we can fix our reduction algorithm by introducing this property, but also we can improve the time complexity, making it
linear.

\begin{figure}
    \centering
    \includegraphics[width=0.2\linewidth]{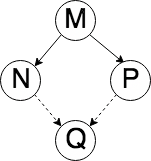}
    \caption{Confluence in a term rewriting system.}
    \label{fig:conf}
\end{figure}

One way to achieve confluence is to attach a single kinship term to any possible path in a family tree.
Observe, that English kinship terms have a specific pattern that we can exploit. All relatives who are distant enough from ego
have the following structure of their kin term:
\begin{align*}
    n^{th} \text{ cousin } m^{th} \text{ times removed}
\end{align*}
In-laws also have their own pattern, where the ending "-in-law" is appended to a valid consanguine kinship term. However, this
applies only to people, who are linked together by only one nuptial bond. For instance, there is no single term for a husband of
ego's wife's sister. These relations can be accounted for by \textit{prefixing} "-in-law" with an ordinal, which shows the number
of marital bonds that one should pass in order to go to such person. Under this representation, last example will receive the term
"brother \textit{twice}-in-law". After generalizing that scheme we get a pattern that looks like this:
\begin{align*}
    \langle \text{Consanguine kinship term} \rangle k^{th} \text{ times-in-law}
\end{align*}

We can also view this as an attribution of a distinct \textit{natural number} to every vertex with ego as an \textit{offset}, thus
imposing a natural ordering on the set of all vertices. This assignment can be made in such a way that reducing a kinship term $n$
will correspond \textit{exactly} to the calculation of $n$ from some arithmetic expression like $5 \cdot (2 + 3) + 4$, thus
providing a \textbf{translation} between the language of all valid arithmetic expressions and our formal language of kinship
$\mathcal{L}$.

However, it is not the topic of this paper, so we are leaving it to the considerations of future researches.

\section{Incorporating Time}
Now the only matter that is left to address is an adequate representation of time. Historically, there are two main approaches
for modelling time: point-based and interval-based. The former treats time as a single continuous line with distinguished points
as specific \textit{events}, and the latter uses \textit{segments} of that line to represent time entries.
The latter method was used in Allen's interval algebra\cite{allen}. For the sake of simplicity we chose the former approach,
because it can easily imitate intervals by treating them as endpoints of a line segment.

Not only we want to be able to express different events by modelling them as points on a line, but also we want to orient ourselves on
that line, i.e. to know where we are, which events took place in the past and which will happen in the future. Thus, we need to
select exactly one point that will stand for the present moment and we call it "now". Then all points to the left will be in the past,
and all point to the right will be in the future. Also, notice that any set with total ordering on it will suffice, because the
continuous nature of a line is redundant in point-based model. Collecting everything together, we have the following formalisation
of time:
\begin{align*}
    \mathcal{M} = \langle T, now, \leqslant \rangle
\end{align*}
Where $T$ is a non-empty set with arbitrary elements, $now \in T$, and $\leqslant$ is a total ordering relation on $T$.

Within this model we can reason about which event comes \textit{before} or \textit{after}, what events took place in the past or
in the future, and so on.

When considering family trees it is necessary to define only five predicates:
\begin{enumerate}
    \item{$Before(x, y)$ is true iff $x < y$.}
    \item{$After(x, y)$ is true iff $x > y$.}
    \item{$During(x, s, f)$ is true iff $s \leqslant x \leqslant f$.}
    \item{$Past(x)$ is true iff $Before(x, now)$.}
    \item{$Future(x)$ is true iff $After(x, now)$.}
\end{enumerate}
Those relations are the basis from which all other operations on $\mathcal{M}$ can be defined. It is also interesting to note
that, since any ordering relation generates a \textit{topology} over its structure, we can speak about time in terms of its
topological properties.

\section{KISP Language Specification}
\subsection{Grammar and Lexical Structure}
Since KISP is a dialect of LISP, it inherits some syntax from the predecessor, but generally it is a new programming language. The
grammar is presented with the help of Bacus-Naur notational technique. For the sake of simplicity, we omit angle brackets and
embolden non-terminal words. A plus, a star sign in a superscript and a question mark have the same meaning as in regular
expressions.
\begin{align*}
    \text{\textbf{term}} &::= \text{\textbf{literal}} | \text{\textbf{lambda}} | \text{\textbf{define}} | \text{\textbf{atom}} |
        (\text{\textbf{term}}^+) \\
    \text{\textbf{lambda}} &::= (\text{lambda (\textbf{reference}}^*\text{) \textbf{term}}) \\
    \text{\textbf{define}} &::= ( \text{define \textbf{reference} \textbf{term}}) \\
    \text{\textbf{reference}} &::= \text{\textbf{word}\{-\textbf{word}\}'?'?} \\
    \text{\textbf{word}} &::= \text{\textbf{letter}}^+ \\
    \text{\textbf{letter}} &::= a | b | ... | z | A | B | ... | Z \\
    \text{\textbf{atom}} &::= * | + | \text{concat} | \text{list} | \text{append} | ... \\
    \text{\textbf{literal}} &::= \text{void} | \text{true} | \text{false} | \text{people} | \text{vacant} | \text{now} |
        \text{\textbf{numeral}} | \text{\textbf{string}} \\
    \text{\textbf{numeral}} &::= \text{-?\textbf{digit}}^* \\
    \text{\textbf{digit}} &::= 0 | 1 | ... | 9 \\
    \text{\textbf{string}} &::= `\text{\textbf{symbol}}^*` \\
    \text{\textbf{symbol}} &::= \text{\textit{any non-blank ASCII symbol}}
\label{al:bnf}
\end{align*}
As we can see from the definition, there are three kinds of terminals in the grammar: literals, references and atom functions,
which are called simply \textbf{atom}s. Literals are instances of primitive types, such as \textit{Numeral}, \textit{String} or
\textit{Boolean}, or special keywords. They stands for the following: "void" represents NULL type, "people" -- a list of
all persons in a family tree, "now" -- the current time entry and "vacant" -- an empty list. References are used as definientia in
"define" terms and as names for parameters in lambda terms. It is possible for a reference to end in a question mark, which means
that it denotes an instance of \textit{Boolean} type. References can we written in so-called \textit{dash case}, so "long-name"
and "very-long-name" are both legal. The only exception are names which start with the dash like "-illegal", they are invalid.

Note that we allow niladic lambdas, so, for instance, this is a valid expression: \texttt{(lambda () 'Hello, World!')}. But at the
same time \texttt{()} is not a well-formed term. We also prohibit "define" terms inside other terms, so this would not work:
\texttt{(+ 2 (define three 3))}. Strings are nested in single quotes. Integers, in KISP we call them "numerals", can start with a
zero and be prefixed by a negative sign.

Here is the complete list of all keywords in KISP: \textbf{true}, \textbf{false}, \textbf{define}, \textbf{lambda},
\textbf{people}, \textbf{now}, \textbf{void}, \textbf{if}, \textbf{vacant}. The rule is that you can use as a reference everything
you want as long as it is not a keyword, so you cannot redefine their standard behaviour, thus a programmer is unable to tamper
with inner workings of the interpreter.

As in all other dialects of LISP, a term \texttt{(f a b c ...)} means the \textit{execution} of a function $f$ with the
specified arguments $f(a, b, c, ...)$. Of course, we can construct and call the higher-order functions as usual:
\texttt{((twice square) 2)} will yield $16$, or \texttt{((compose inc inc) 0)} which just prints $2$.
\subsection{Query Examples}
In this section we will demonstrate how one can use KISP to perform various queries in a genealogical tree. Particularly, we
focus our attention on statements that express kinship terms.

Let's start with a simple task of selecting people based on a certain boolean condition. Suppose we want to query only those, who
have at least one child. This can be accomplished as follows:
\begin{verbatim}
(filter (lambda (p) (< 0 (count (children p)))) people)
\end{verbatim}
Here we iterate through the list of all people in a tree and take only those, on who defined lambda predicate evaluated to
\texttt{true}. The number of children for a particular person is calculated by counting elements of the list \texttt{(children
p)}.

The next task is to select all husbands, that is, all men who are married. This can be done in two ways: either select only
males and then discard all bachelors, or combine the two operations together in a single boolean predicate using \texttt{and}
clause:
\begin{verbatim}
(filter (lambda (p) (not (= void (spouse p))))
        (filter (lambda (p) (= 'MALE' (attr p 'sex')))
        people))
(filter (lambda (p) (and (= 'MALE' (attr p 'sex'))
                        (not (= vacant (spouse p)))))
people)
\end{verbatim}
The advantage of the second approach is that the list \texttt{people} will be iterated only once.

Now to the more advanced queries; suppose that the term \texttt{ego} stands for the user's node in an ancestry, and he wants to
know how many cousins he has:
\begin{verbatim}
(define parents (lambda (p) (join (mother p) (father p))))
(define cousins
(lambda (p) (children (children (parents (parents p))))))
(- (count (cousins ego)) 1)
\end{verbatim}
This is where the expressive power of KISP truly comes into play. Although \texttt{cousins} is not a standard KISP function, we
can easily implement it using kinship framework of KISP, which successfully utilizes the structure of natural kinship terms.
Moreover, notice how the function \texttt{parents} is expressed. Since a parent is either a mother or a father, it corresponds to
the formal kinship term $(mother \vee father)$, which is implemented as a \textit{join} of two or more lists. And because every
cousin is a grand-child of one's grandparents, it corresponds to:
\begin{align*}
(son \vee daughter)(son \vee daughter)(mother \vee father)(mother \vee father)
\end{align*}
The last decrement was made because in this scheme the \texttt{ego} itself will be included to the resulting list.

Finally, temporal queries can be expressed with the help of the type \textit{Date}. For instance, if we need to know, who, among our
relatives, was born during the WWII, we just need to evaluate:
\begin{verbatim}
(define WWII-start (date '01.09.1939'))
(define WWII-end (date '02.09.1945'))
(filter (lambda (p) (during (attr p 'birthdate') WWII-start WWII-end))
people)
\end{verbatim}
The type \textit{Date} provides all the necessary functions for working with temporal information.

Because KISP is Turing-complete and inherits LISP's capabilities of meta-programming, one can easily extend it with any
functionality that one wants.

\section{Conclusion}
In this work we designed a new programming language KISP for effectively navigating and querying temporal family trees. We described a formal mathematical model of traditional kinship, on which KISP is based.
Additionally, we tackled the problem of term reduction and discussed the possibilities for achieving confluence.

There are three requirements that we want our language to satisfy:
\begin{enumerate}
    \item{\textbf{Expressiveness}. The language should allow for any possible consanguine as well as affinal relations to be
        described.}
    \item{\textbf{Speed}. The response time must not exceed the standard for an interpreted language.}
    \item{\textbf{Simplicity}. Language should be able to express natural kinship and temporal terms as straightforward as possible.}
\end{enumerate}
Here the phrase "response time" stands for the time passed between the start and finish of a programs evaluation.
The last quality is what truly distinguishes our approach from the rest, allowing for the most obvious representation of
genealogical and temporal information. We are certain that our presented solution fully covers every one of them.

\section{Future Work}
However, there are some topics yet left to tackle in the area of kinship and genealogy management. On the theoretical side, there
is a problem of total term reduction and formal language enrichment. It is also interesting to shift attention to other languages
and cultures with different kinship structures, such as Russian or Hawaiian. The constructed formalism can be considered from the
algebraical side, focusing on its many mathematical properties as a special type of an algebraic system.

On the practical side, one can consider to improve the virtual assistant component. Besides already mentioned Voice Generation \&
Recognition technology, it can be made context-aware, which will increase its intelligence. Additionally, the family ontology can
be enhanced to incorporate information about divorces and deaths. The performance of the interpreter can be significantly improved
by introduction of \textit{Just-in-Time} compilation.

Another important step towards improvement of existing system is addressing its current limitations, such as static nature of our
genealogical database and execution time of KISP queries. It's no doubt that the schema-less approach is much more versatile,
because it does not depend on traditions and customs of a particular culture.

Further advancement may also include new data types and standard functions for KISP language. Specifically, it is beneficial to
add a \texttt{char} type that represents individual characters in a string. Another useful feature is support for \textit{variadic
lambdas} and \textit{closures}, which will significantly increase the versatility of KISP.

Moreover, one can also consider including capabilities for a logical reasoning into KISP. They will be applicable for inferring
implicit time constraints for events, whose exact date is unknown. For instance, if we are uninformed about a birthday of a
person, but we do know his parents and his children birthdays, we can justifiably bound this missing date to a specific time
interval.

\section*{Appendix}
\begin{algorithm}[caption={Kinship Term Reduction}, label={algo:red}]
input: kinship term $t$.
note: $\omega$ is a dictionary of kinship terms.
note: "leftPart(t, u)" and "rightPart(t, u)" return the sub-term of $t$
note: from the left of sub-term $u$ or from the right respectively.
note: Function "subterm(t, i, j)" returns the
note: sub-term of the kinship term $t$ between indices $i$ and $j$.
output: reduced kin term.
function shorten(t)
begin
    maxShortenableSubterm $\gets$ empty
    currentSubterm $\gets$ empty
    for i $\gets$ 0 to length(t) do
    begin
        for j $gets$ length(t) - i to 0 do
        begin
            currentSubterm $\gets$ subterm(t, i, j)
            if length(currentSubterm) $>$ length(maxShortenableSubterm)
                and $\omega$(currentSubterm) is not empty
            then
                maxShortenableSubterm = currentSubterm
        end
    end
    return shorten(leftPart(t, maxShortenableSubterm))
            $\cdot$ $\omega(maxShortenableSubterm)$
            $\cdot$ shorten(rightPart(t, maxShortenableSubterm))
end
\end{algorithm}
\bibliographystyle{IEEEtran}
\bibliography{annot}
\end{document}